\documentclass[12pt]{article}

\usepackage{cite}
\usepackage{amsmath,amsthm,amssymb}
\usepackage{tikz,tikz-cd}
\usepackage{graphicx}
\usepackage{mathrsfs}
\usepackage[scr=boondox]{mathalfa} 
\usepackage{listings}

\newtheorem{theorem}{Theorem}[section]
\newtheorem{corollary}{Corollary}[theorem]

\newtheorem{proposition}[theorem]{Proposition}

\theoremstyle{definition}

\theoremstyle{definition}

\theoremstyle{definition}
\newtheorem{axiom}{Axiom}
\theoremstyle{definition}
\newtheorem{definition}{Definition}[section]

\theoremstyle{definition}
\newtheorem*{remark}{Remark}

\newcommand{\ob}[1]{{\text{ob}\!\left(#1\right)}}
\renewcommand{\hom}[1]{{\text{hom}\!\left(#1\right)}}

\newcommand{\Ep}{{\textbf{\text{Ep}}}}

\author{Sotirios Henning}
\title{A Monadic Calculus with Episodic Flows}

\begin{document}
\maketitle

\begin{abstract}
We define computational atoms named ``actions" equipped primarily with three operations: reduction, collection, and inspection. We show how actions can be used for decision-making algorithms from simple axioms. We describe the encodings of typical data structures as actions, and provide a method of analysis for algorithms on the basis of data mutation.
\end{abstract}

\section{Introduction}
We recognise the challenge of avoiding undefined behaviour in computer systems, and seek to construct a framework of computation by which the change of data is itself quantifiable data. More generally, we seek to expand the notion of a monad by understanding computation in terms of nested actions, and what defines the success or failure thereof. We expect deterministic algorithms to naturally follow from a metric for change in information.

It is helpful to delineate axioms which precurse the formal and nuanced abstractions we wish to later describe. Recall that the notions of ``success" and ``failure" are eminent in the desired model of computation.

\begin{axiom} \label{ax:Succeeds}
If we observe a particular (not general) flow of actions such that a final action occurs that yields a desirable outcome, the flow of actions is successful, or succeeds.
\end{axiom}

\begin{axiom} \label{ax:Fails}
Correspondingly, if we observe a particular flow of actions such that a final action occurs that yields a pessimal outcome, the flow of actions is unsuccessful, or fails.
\end{axiom}

\begin{definition} \label{def:Nul}
Let the category $\Ep$ be defined by the following diagram. The elements $\omega,\omega^*\in\ob{\Ep}$ are called the \textbf{episodics}.
\end{definition}

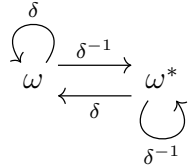
\begin{center}
\begin{tikzcd}
\omega \arrow[r, "\delta^{-1}", shift left=2] \arrow["\delta"', loop, distance=2em, in=125, out=55] & \omega^* \arrow[l, "\delta", shift left=2] \arrow["\delta^{-1}"', loop, distance=2em, in=305, out=235]
\end{tikzcd}
\end{center}

\begin{axiom} \label{ax:SuccessAndFailure}
Consider a morphism $f:X\rightarrow\ob\Ep$. Given $x\in X$, if $f(x)=\omega^*$, then $f$ succeeds under inspection of $x$, otherwise if $f(x)=\omega$, then $f$ fails under inspection of $x$, as per axioms \ref{ax:Succeeds} and \ref{ax:Fails} respectively.
\end{axiom}

\begin{axiom} \label{ax:NullIsSuccess}
Doing nothing always succeeds.
\end{axiom}

\begin{remark}
To construct a model of computation which accepts notions of success and failure, we require a category which accepts $\Ep$ as objects at the base of a cone of which every other object is a vertex.
\end{remark}

\begin{definition} \label{def:Episode}
Given a category $\mathcal C$, its \textbf{episode} is the set given by equation \ref{eq:Episode}.
\end{definition}

\begin{align} \label{eq:Episode}
E(\mathcal C)=\{\omega\}\cup\{\omega^*\}\times \ob{\mathcal C}
\end{align}

\begin{definition} \label{def:Reduction}
The morphism $\delta:E(\mathcal C)\rightarrow\ob{\Ep}$ is the \textbf{reduction} of the episode of $\mathcal C$, given by equations \ref{eq:Reduction_1} and \ref{eq:Reduction_2}. Note that $\delta$ is called the \textbf{reductor} operator.
\end{definition}

\begin{definition} \label{def:ReducibleProperty}
A category $\mathcal C$ is \textbf{reducible} if $\Ep$ is a full subcategory of $\mathcal C$.
\end{definition}

\begin{definition} \label{def:InspectibleProperty}
A reducible category $\mathcal C$ is \textbf{inspectible} if for all $\alpha\in\ob{\mathcal C}$, there exists at least one of each $f_\alpha,f_\alpha^{-1}\in\hom{\mathcal C}$ such that $f_\alpha=\omega^*$ and $f_\alpha^{-1}=\omega$.
\end{definition}

\begin{remark}
Consider $\mathbb N$, the set of natural numbers, to include $0$.
\end{remark}

While monads are considered to enjoy a combinator function, if we consider an element of an episode to be a monad of the form

\texttt{Maybe := Just(flow) | None} 
then the reductor on this monad will yield whether the monad is a \texttt{Just} or a \texttt{None}, and a combinator should yield \texttt{flow} if the monad is in the \texttt{Just} form. However, leaving the case of the \texttt{None} undefined is erroneous. Contextually, comonads are relevant when there is uncertainty of this nature, but techniques that employ comonads require conditions which are similarly uncertain, without loss of generality. We will explore the formal notions of such erroneous behaviour in section \ref{sec:Entropy}.

\begin{align}
\label{eq:Reduction_1}
\delta\left(\omega^*,\varphi\right)&=\omega^* \\
\label{eq:Reduction_2}
\delta\omega&=\omega\;\;\text{per definition \ref{def:Nul}}
\end{align}

Moggi \cite{MOGGI199155} considers a monad over a category to be a specification of a Kleisli triple, which is the categorical definition, and, currently, the usual construction for functional programming. We take an approach that similarly reduces the notion of a monad to a stronger definition that should approach a useful concept in the ecosystem given by this particular calculus.

\begin{definition} \label{def:CollectibleProperty}
A reducible category $\mathcal C$ is \textbf{collectible} if, given a commutative operation with the below signature, $\alpha\in\varphi$ if and only if equation \ref{eq:IndexOperationSuccess} holds, and $\alpha\notin\varphi$ if and only if equation \ref{eq:IndexOperationFailure} holds.
\end{definition}

\begin{align*}
\circ:\left(\mathbb N\times\{\omega\}\right)\times\left(\mathbb N\times\{\omega^*\}\right)&\times\ob{\mathcal C}\rightarrow E(\mathcal C) 
\end{align*} \begin{align}
\label{eq:IndexOperationSuccess}
x\omega\circ_{\varphi}y\omega^*&=(\omega^*,\alpha) \\
\label{eq:IndexOperationFailure}
x\omega\circ_{\varphi}y\omega^*&=\omega
\end{align}

\begin{remark}
It is convenient to use double-index notation such that $(\omega^*,\varphi_{pq})=p\omega\circ_\varphi q\omega^*$. For a flow such that there does not exist a natural number $n>0$ such that either $\omega\neq n\omega\circ_\varphi q\omega^*$ or $\omega\neq p\omega\circ_\varphi n\omega^*$, single-index notation is fine, e.g. $(\omega^*,\varphi_{i})=0\omega\circ_\varphi i\omega^*$.
\end{remark}

\begin{definition} \label{def:EpisodicProperty}
A reducible category $\mathcal C$ is \textbf{episodic} if $E(\mathcal C)\subseteq\ob{\mathcal C}$.
\end{definition}

\begin{definition} \label{def:ActionCategory}
A category is an \textbf{action category} if it is inspectible and collectible. Note that it follows that the category is reducible.
\end{definition}

\begin{definition} \label{def:Action}
An object in an action category is an \textbf{action}.
\end{definition}

\begin{definition} \label{def:Flow}
An action $\alpha$ such that $\alpha\notin\ob{\Ep}$ is a \textbf{flow}. Note that every flow forms the vertex of a cone of which the base is $\Ep$.
\end{definition}

Given the morphism $\left<*\right>:\ob{\mathcal A}\rightarrow E(\mathcal A)$, $\delta\!\left<\alpha\right>$ fulfills axiom \ref{ax:SuccessAndFailure} if we consider the operation to yield the status of $\alpha$'s success.

\begin{definition} \label{def:EpisodicFlowCategory}
An episodic action category $\mathcal E$ is an \textbf{episodic flow category} if, for every action $\alpha\in\ob{\mathcal E}$, there is a morphism $\left<\alpha\right>\in E(\mathcal E)$, and a morphism $\delta\left<\alpha\right>\in\ob{\Ep}$.
\end{definition}

\begin{definition} \label{def:EpisodicFlow}
A flow in an episodic flow category is an \textbf{episodic flow}.
\end{definition}

\begin{proposition} \label{prop:EpFlowCat}
A reducible subcategory of $\mathbf{Cat}$, $\mathcal E$, whose each flow $\varphi$ is a category with the following structure, is an episodic flow category if each $\varphi_{pq}\in\mathcal E$.
\end{proposition}

\begin{center}
\begin{tikzcd}
\varphi_{11} \arrow[r, "\sigma_{11}"] \arrow[d, "\theta_{11}"]            & \varphi_{12} \arrow[r, "\sigma_{12}"] \arrow[d, "\theta_{12}"]           & \cdots \arrow[r, "{\sigma_{1,n-1}}"] \arrow[d]                         & \varphi_{1n} \arrow[d, "\theta_{1n}"] \arrow[r, "\Sigma_1"]           & {(\omega^*,\varphi_{1n})} \\
\varphi_{21} \arrow[r, "\sigma_{21}"] \arrow[d, "\theta_{21}"]            & \varphi_{22} \arrow[r, "\sigma_{22}"] \arrow[d, "\theta_{22}"]           & \cdots \arrow[r, "{\sigma_{2,n-1}}"] \arrow[d]                         & \varphi_{2n} \arrow[d, "\theta_{2n}"] \arrow[r, "\Sigma_2"]           & {(\omega^*,\varphi_{2n})} \\
\vdots \arrow[r] \arrow[d, "{\theta_{m-1,1}}"]                            & \vdots \arrow[r] \arrow[d, "{\theta_{m-1,2}}"]                           & \ddots \arrow[r] \arrow[d]                                             & \vdots \arrow[d, "{\theta_{m-1,n}}"] \arrow[r, "\Sigma_k"]            & {(\omega^*,\varphi_{kn})} \\
\varphi_{m1} \arrow[r, "\sigma_{m1}"] \arrow[rrd, "\Theta_1" description] & \varphi_{m2} \arrow[r, "\sigma_{m2}"] \arrow[rd, "\Theta_2" description] & \cdots \arrow[r, "{\sigma_{m,n-1}}"] \arrow[d, "\Theta_k" description] & \varphi_{mn} \arrow[r, "\Sigma_m"] \arrow[ld, "\Theta_n" description] & {(\omega^*,\varphi_{mn})} \\
                                                                          &                                                                          & \omega                                                                 &                                                                       &                          
\end{tikzcd}
\end{center}

\begin{proof}
First, we show $\mathcal E$ is inspectible. Consider that for every morphism $I=\Theta_i\circ\cdots$, $\delta\circ I:E(\mathcal E)\rightarrow\{\omega\}$, and for every morphism  $I^*=\Sigma_i\circ\cdots$, $\delta\circ I^*:E(\mathcal E)\rightarrow\{\omega^*\}$. Therefore, either morphism $\delta\circ I$ or $\delta\circ I^*$ maps $\ob{\mathcal E}\rightarrow\ob{\Ep}$.

Next, we show that $\mathcal E$ is collectible. Allow $(\omega^*,\varphi_{pq})=p\omega\circ_\varphi q\omega^*$.

Finally, we show that $\mathcal E$ is episodic. This is trivial, as each $\Sigma_i$ maps $\ob{\varphi}\rightarrow\{\omega^*\}\times\ob{\varphi}$ and each $\Theta_i$ maps $\ob{\varphi}\rightarrow\{\omega\}$.

These properties satisfy the requirement for an episodic flow category.
\end{proof}

Thus far, we have only seen the initial definitions of this model, and it is perhaps still unclear which notions are analogous to monads and comonads. Intuitively, elements of an episode might seem to be monadic in nature, as episodes would provide an interface essentially equivalent to monads given certain necessary natural transformations. However, although there is a similarity between episodes and monads, it is actually the episodic flow that we find to be analogous to monads, and the operation $\left<\varphi\right>$ that yields an element of an episode is an analogue of a monadic operation.

\section{Theoretical semantics}
It would be helpful to work within the category defined partially in proposition \ref{prop:EpFlowCat}. In fact, the theory of episodic flows that is described in this paper will apply techniques exclusively on this category, in which flows are categories notated as matrices. This notation allows semantics with operations that are common for matrices.

This particular episodic flow category instantiates great utility on account of its simplicity, wherein the category maintains the necessary properties based on the categorical structure of its objects. Indeed, when referring to flows going forward, \textit{we are speaking strictly about episodic flows in this category}.

\begin{definition}
The episodic flow category $\Phi$ is the category described in proposition \ref{prop:EpFlowCat}, equipped with morphisms

$\left<*\right>:\ob{\Phi}\rightarrow E(\Phi)$, $\delta:\ob{\Phi}\rightarrow\ob{\Ep}$, and $\delta^{-1}:\ob{\Phi}\rightarrow\ob{\Ep}$, with the usual definitions.
\end{definition}

Observe that it is natural to represent flows in $\Phi$ as an object similar to a lattice. Indeed, there is a new immediate options for such a representation: a structure like a bounded lattice, by which we exclude the extremities given by each $\Sigma$- and $\Theta$-morphism, allowing the initial object to be an infinimum and the final object to be a supremum.

We denote these objects in an established notation, by describing an $m\times n$ matrix for the described lattice-like structure.

\begin{definition} \label{def:ProcessMatroid}
A matrix $\Lambda$ is a \textbf{process lattice} if there exists a flow $\varphi$ such that every $\Lambda_{ij}=\varphi_{ij}$. Moreover, $\Lambda$ is said to model $\varphi$, written $\Lambda\cong\varphi$.
\end{definition}

\begin{align*}
\Lambda&=\begin{pmatrix}
\varphi_{11}&\varphi_{12}&\cdots&\varphi_{1n}\\
\varphi_{21}&\varphi_{22}&\cdots&\varphi_{2n}\\
\vdots&\vdots&\ddots&\vdots\\
\varphi_{m1}&\varphi_{m2}&\cdots&\varphi_{mn}
\end{pmatrix}\\
&\cong\varphi
\end{align*}

This notation maintains some meaningful matrix operations if we consider $(\omega^*,\Lambda_{ij})=i\omega\circ_\varphi j\omega^*$. We can compute $\left<\varphi\right>$ by applying a recursive algorithm on $\Lambda$.

\begin{definition}
The set of all process lattices given by $\Phi$ is written $\mathbb L$.
\end{definition}

\begin{definition} \label{def:Inspection}
The operation $\left<\alpha\right>$ on an action $\alpha$ is called the \textbf{inspection} of $\alpha$. It is computed as follows.

First let $\left<\omega^*\right>=(\omega^*,\omega)$, and $\left<\omega\right>=\omega$.
If $\alpha$ is a flow compute the following steps in order, considering $\Lambda\cong\alpha$ to be an $m\times n$ process lattice.

\begin{enumerate}
  \item Let $p,q=1$.
  \item If $\delta\left<\alpha_{pq}\right>=\omega$, increment $p$. If $\delta\left<\alpha_{pq}\right>=\omega^*$, increment $q$. Repeat this step until either $p=m+1$ or $q=n+1$.
  \item If $q=n+1$, $\left<\alpha\right>=(\omega^*,\alpha_{p,q-1})$.
  \item If $p=m+1$, $\left<\alpha\right>=\omega$.
\end{enumerate}

Inspection is also considered a valid operation on process lattices, such that $\left<\Lambda\right>=\left<\varphi\right>$ for any $\Lambda\cong\varphi$.

If $\Lambda$ is an empty matrix, this operation implies $\left<\Lambda\right>$ takes the form $(\omega^*,\cdot)$. However, there is no element in $\Lambda$ that gives the inspection a value, so we allow $\left<\alpha\right>=\left<\Lambda\right>=(\omega^*,\omega)$. This fulfills axiom \ref{ax:NullIsSuccess}, as an empty flow effectively ``does nothing" succesfully.
\end{definition}

\begin{proposition}
For a flow $\varphi$, $\left<\varphi\right>=(\omega^*,\omega)$ if and only if there are no elements in $\varphi$.
\end{proposition}

\begin{proof}
The proof is by contradiction. Notice that there are two cases in which $\left<\varphi\right>$ may equal $(\omega^*,\omega)$: $\varphi$ contains no elements, or the last value reached during inspection was $\omega$.

Suppose the last value reached during inspection was $\omega$. This implies that, during that step, inspection occurs on the last column of some $\Lambda\cong\varphi$. Moreover, it would have been evaluated that $\delta\left<\omega\right>=\delta\omega=\omega$, which would not have incremented $q$, but rather $p$, implying that $p=m+1$. This further implies $\left<\varphi\right>=\omega$, which contradicts our assumption. Therefore, $\varphi$ is necessarily empty.
\end{proof}

\begin{remark}
An empty flow $\varphi$ should be taken to contain only one object $(\omega^*,\omega)\in\ob{\varphi}$ and $\hom{\varphi}=\emptyset$. This is consistent with the definition of inspection.
\end{remark}

\begin{remark}
The only two solutions to $\left<\alpha\right>$ in $\Phi$ are $\alpha=(\;)$ and $\alpha=\omega^*$.
\end{remark}

\begin{definition}
Let $F(\Lambda)=p^\prime-1$ and $S(\Lambda)=q^\prime-1$, where $p^\prime$ and $q^\prime$ are the last values of $p$ and $q$ reached during inspection, respectively.
\end{definition}

\begin{definition} \label{def:EpInverse}
The \textbf{episodic inverse} of an action, written $\alpha^*$ is defined such that
\begin{itemize}
  \item If $\alpha$ is a flow, given the process lattice $\Lambda\cong\alpha$, $\alpha^*\cong\Lambda^*$, where $\Lambda^*_{ij}=(\Lambda^T_{ij})^*$.
  \item Otherwise, $(\omega)^*=\omega^*$ and $(\omega^*)^*=\omega$.
\end{itemize}
\end{definition}

\begin{proposition}
For any $m\times n$ process lattice $\Lambda$, $(\delta\left<\Lambda\right>)^*=\delta\left<\Lambda^*\right>$ if $m,n>0$.
\end{proposition}

\begin{proof}
Notice that $\Lambda_{pq}=\Lambda^*_{qp}$. We compute the necessary steps to find $\left<\Lambda\right>$ and check that either $p^\prime=m+1$ or $q^\prime=n+1$, where $p^\prime$ and $q^\prime$ are the final values assumed by the indices $p$ and $q$ during inspection. In the case that $p^\prime=m+1$, the last action in $\Lambda$ is necessarily on row $m$, which corresponds to the column $m$ in $\Lambda^*$. This implies $\delta\left<\Lambda\right>=\omega$ and $\delta\left<\Lambda^*\right>=\omega^*$. Similarly, in the case that $q^\prime=n+1$, the last action in $\Lambda$ is necessarily on column $n$, which corresponds to the row $n$ in $\Lambda^*$. This implies $\delta\left<\Lambda\right>=\omega^*$ and $\delta\left<\Lambda^*\right>=\omega$. 

Therefore, $\delta\left<\Lambda\right>=\omega^*\Leftrightarrow
\delta\left<\Lambda^*\right>=\omega$ and $\delta\left<\Lambda\right>=\omega\Leftrightarrow
\delta\left<\Lambda^*\right>=\omega^*$, and it immediately follows that $(\delta\left<\Lambda\right>)^*=\delta\left<\Lambda^*\right>$.
\end{proof}

\begin{corollary}
It follows that $S(\Lambda)=F(\Lambda^*)$ for a non-empty process lattice $\Lambda$.
\end{corollary}

We describe the composition of process lattices to observe sequential execution as the inspection of nested flows.

\begin{definition}
If $\Lambda_2$ is a process lattice given by the map $\Lambda_2:\mathbb L\rightarrow\mathbb L$, the \textbf{composition} of $\Lambda_1$ with $\Lambda_2$ is written $\Lambda_2(\Lambda_1)=\Lambda_1\setminus\Lambda_2$. Composition is left-associative, so $\Lambda_1\setminus\Lambda_2\setminus\Lambda_3=(\Lambda_1\setminus\Lambda_2)\setminus\Lambda_3$.
\end{definition}

\subsection{Encoding}
It is necessary for a model of computation to be capable of encoding numbers. In general, it is desirable to find an algorithm for encoding information as an action.

\begin{definition}
A \textbf{lift} $\Delta_X:X\rightarrow\ob{\Phi}$ of an object in $X$ is an isomorphism ($\Delta_X(x+y)=\Delta_X(x)+\Delta_X(y)$ for all $x,y\in X$ and all operations $+$ on $X$). A \textbf{colift} $\Delta^*_X:\ob{\Phi}\rightarrow X$ of an object in $X$ is also an isomorphism such that $\Delta_X^*(\Delta_X(x))=x$.
\end{definition}

\begin{definition}
An element of the episode $E(\Phi)$ of the form $(\omega^*,\alpha)=\rho$ is isomorphic to a process lattice given by $\bar\rho=\begin{pmatrix}
\omega^*&\alpha
\end{pmatrix}\cong\Delta_{E(\Phi)}(\rho)$
\end{definition}

\begin{definition}
Given $\bar\rho=\begin{pmatrix}
\omega^*&\alpha
\end{pmatrix}$, $\hat\rho=\alpha$.
\end{definition}

\subsection{Arithmetic}\label{sec:Arithmetic}

Because the functions $S$ and $F$ allow mapping from a process lattice to a natural number, we can find a colift to the ring of integers immediately. Let $\Delta_{\mathbb Z}^*$ be given as follows, where $\Lambda\cong\varphi$.

\begin{align*}
\Delta_{\mathbb Z}^*(\varphi)&=\begin{cases}
S(\Lambda)&F(\Lambda)=0\\
-F(\Lambda)&\text{otherwise}
\end{cases}
\end{align*}

Consider that a process lattice of this form must naturally always succeed, lest we assume an erroneous result for coming techniques. This leads to a natural construction of a lift from the ring of integers: $0$ maps to the empty flow, positive integers map to a flow given by a process lattice that is written as a row matrix of that many $\omega^*$, and negative integers map to a flow given by a process lattice that is written as a column matrix of that many $\omega$, with one row of a single $\omega^*$ at the bottom, to ensure success upon inspection, as such with these examples:

\begin{align*}
\Delta_{\mathbb Z}(2)&\cong\begin{pmatrix}
\omega^*&\omega^*
\end{pmatrix}\\
\Delta_{\mathbb Z}(-5)&\cong\begin{pmatrix}
\omega\\\omega\\\omega\\\omega\\\omega\\\omega^*
\end{pmatrix}
\end{align*}

\begin{definition}
The process lattice $\Lambda\cong\varphi$ given by $\Delta_{\mathbb Z}(n)=\varphi$ is written as a boldface $\mathbf n$ for the numeral represented by $n$.
\end{definition}

From these definitions, arithmetic operations on (within the ring of integers) on process lattices are easily described in terms of such isomorphisms.

Given the collection operation, it is possible to describe arithmetic for process lattices of lifted integers. Given $a,b\geq0$,

\begin{align*}
\mathbf a+\mathbf b&\cong\begin{pmatrix}
\mathbf a_1&\cdots&\mathbf a_a
&\mathbf b_1&\cdots&\mathbf b_b
\end{pmatrix}\\
\mathbf a-\mathbf b&\cong\begin{pmatrix}
\mathbf a_1&\cdots&\mathbf a_{a-b}
\end{pmatrix}&\text{where }a\geq b\\
\mathbf a-\mathbf b&\cong\begin{pmatrix}
\mathbf a_1&\cdots&\mathbf a_{b-a}&\omega
\end{pmatrix}^*&\text{where }a<b\\
-\mathbf a&\cong\begin{pmatrix}
\mathbf a_1\\\vdots\\\mathbf a_a\\\omega^*
\end{pmatrix}
\end{align*}

Let the colift of rational numbers be

\begin{align*}
\Delta_{\mathbb Q}^*(\varphi)&=\begin{cases}
0&S(\Lambda)=F(\Lambda)=0\\
\frac{S(\Lambda)-1}{\Delta_{\mathbb Z}^*(\alpha)}&\text{otherwise}
\end{cases}\\\\
&\text{where }\alpha\cong\widehat{\left<\Lambda\right>}
\end{align*}

With this colift, the lift of natural numbers should see each process lattice as an $n+1\times 1$ row matrix of $\omega^*$, with the final column being a lift of an integer. By taking the quotient $\frac{n}{\Delta_{\mathbb Z}^*(\alpha)}$, it is trivial to see that this colift is surjective.

\subsection{Formal Logic}

We can compute formal logic by treating $\omega^*$ as true and $\omega$ as false. We see an encoding of formal logic through process lattices with the following inspection-reduction operations:

\begin{align*}
a\vee b&\mapsto\delta\left<\begin{matrix}
a&b\\
b&\omega^*
\end{matrix}\right>\\
a\wedge b&\mapsto\delta\left<\begin{matrix}
a&b
\end{matrix}\right>\\
\neg a&\mapsto a^*\\
a\rightarrow b&\mapsto\delta\left<\begin{matrix}
a^*&b\\
b&\omega^*
\end{matrix}\right>\\
&=\delta\left<\begin{matrix}
a&b\\
\omega^*&a^*
\end{matrix}\right>\\
a\leftrightarrow b&\mapsto\delta\left<\begin{matrix}
a&b\\
b^*&a^*
\end{matrix}\right>
\end{align*}

The last operation is particularly useful, as it can be used to show that two flows are equal within the theory of episodic flows.

\begin{definition}
For two $m\times n$ process lattices $\Lambda,\Lambda^\prime$, the episodic $\text{eqpl}(\Lambda,\Lambda^\prime)$ is given by equation \ref{eq:EqualPL}. For an $a\times b$ process lattice $\Lambda^{\prime\prime}$ such that either $a\neq m$ or $b\neq n$, $\text{eqpl}(\Lambda,\Lambda^{\prime\prime})=\omega$.
\end{definition}

\begin{align} \label{eq:EqualPL}
\text{eqpl}(\Lambda,\Lambda^\prime)&=\delta\left<M\right>
\end{align}

where $M$ is a $1\times mn$ process lattice given recursively by

\begin{align*}
M_{im+j}&=\begin{cases}
\Lambda_{ij}\leftrightarrow\Lambda^\prime_{ij}&\Lambda_ij,\Lambda^\prime_{ij}\in\ob{\Ep}\\
\text{eqpl}(\Lambda_{ij},\Lambda^\prime_{ij})&\text{otherwise}
\end{cases}
\end{align*}

It is trivial to see $\text{eqpl}(\Lambda,\Lambda^\prime)=\omega^*$ when $\Lambda=\Lambda^\prime$ and $\text{eqpl}(\Lambda,\Lambda^\prime)=\omega$ when $\Lambda\neq\Lambda^\prime$.

\subsection{Typing}

\begin{definition} \label{def:MorphismStructure}
The \textbf{morphism structure} of an $m\times n$ process lattice $\Lambda\cong\varphi$, written $\mu\Lambda$, is a $2m\times n$ matrix such that every morphism in $\hom{\varphi}$ is in its respective place, such that
\begin{itemize}
  \item Every $\sigma_{ij},\theta_{ij}$ is on column $j$.
  \item Every $\sigma_{ij}$ is on row $2i-1$.
  \item Every $\theta_{ij}$ is on row $2i$.
  \item Every $\Sigma_{i}$ is on row $2i-1$, column $n$.
  \item Every $\Theta_{j}$ is on row $m$, column $j$.
\end{itemize}
\end{definition}

\begin{definition}
The \textbf{partial morphism structure} of an $m\times n$ process lattice $\Lambda\cong\varphi$, written $\pi\Lambda$, is a $(2m-1)\times (n-1)$ matrix such that every morphism in $\hom{\varphi}$ except for $\Sigma$- and $\Theta$-morphisms is in its respective place, such that
\begin{itemize}
  \item Every $\sigma_{ij},\theta_{ij}$ is on column $j$.
  \item Every $\sigma_{ij}$ is on row $2i-1$.
  \item Every $\theta_{ij}$ is on row $2i$.
\end{itemize}
\end{definition}

Notice that the fact that a morphism structure contains the morphisms between actions in $\varphi$ implies that two flows with similar conditions for success and failure will have the same morphism structure.

\begin{definition} \label{def:Type}
Two flows such that $\varphi\cong\Lambda$ and $\varphi^\prime\cong\Lambda^\prime$, $\varphi$ and $\varphi^\prime$ \textbf{are of the same type} if $\mu\Lambda=\mu\Lambda^\prime$. This is written $\varphi\sim\varphi^\prime$.
\end{definition}

\begin{definition}
Two flows such that $\varphi\cong\Lambda$ and $\varphi^\prime\cong\Lambda^\prime$, $\varphi$ \textbf{precedes} $\varphi^\prime$ if the morphism structure of the latter can be partitioned to include the partial morphism structure of the former. This is written $\varphi\vdash\varphi^\prime$.
\end{definition}

Consider that two process lattices, $\Lambda$ and $\Lambda^\prime$, may not be of the same type even if $\Lambda=\Lambda^\prime$. Although they may both yield the same result upon inspection, further analysis may reveal that certain morphisms within $\mu\Lambda$ and $\mu\Lambda$ have differing domains, codomains, or images.

\subsection{Invariance}
Flows are prescriptively process-like entities, although we have shown instances of encodings of data to flows from which resource-centric analyses directly benefit. It would be fruitful to show encodings of robust datatypes; namely, lifts of abstract datatypes would nearly complete a modern framework of computation, with substantial utility in computer programming. Union- and intersection-type data structures should enjoy lifts which treat the selection of data therein as flows.

A common technique in the analysis of algebraic union types, taken often for granted, is to assign each variant of the union structure a relatively unique integer or natural number, in the case of enumerated types.

To holistically describe complex data structures, it is necessary to consider the interface with which data is accessed. With regard to union data structures, which should be constructed in terms of variants thereof, we can consider an object of that structure with a type $T$ to be a triple $(\chi_T,\chi_V,v)$, where the first object is an invariant value unique to the type $T$, the second is an invariant value unique to the variant type $V$, and the last is a value $v\in V$. As a process lattice, this value might take the form $(\:\chi_T\;\chi_V\;v\:)$.

With regard to intersection data structures, we should consider the object to be a mapping $f$ from such an invariant $\chi_V$ to a value $v$, such that a value of type $T_k$ in a structure given by a Cartesian product $T_1\times T_2\times\cdots\times T_n$ is $f(\chi_{T_k})$.

It is now only necessary to find a technique with which to match a particular object to its appropriate ``type," which is doable by finding the respective invariant. A union $T_1\cup T_2\cup\cdots\cup T_n$ or intersection $T_1\times T_2\times\cdots\times T_n$ should have an invariant uniquely derived from the invariants $\chi_{T_1},\chi_{T_2},$ et cetera, and there should be some invariants that are initial, in that they are unique to types which are neither union nor intersections.

An arithmetically trivial way to accomplish this would be to select each initial type and let its invariant be a unique pair of unique prime numbers. The invariant of a union would then be a product of the first prime number in each pair, which is known to be a unique natural number. The invariant of an intersection would then be a product of the second prime number in each pair. This convention might be convenient based on established number theory, but it would be memory-inefficient in a digital computer, as large or nested structures would have invariants with values much greater than a real digital computer could store in volatile memory.

Another technique might be to simply let an initial invariant take the form of a unique symbol, and allow a union type to be a pair $(U,X)$ where $U$ is simply a symbol that denotes a union data structure, and $X$ is the set of all invariants unified by the data structure. Similarly, the invariant $(N,X)$ of an intersection type lets $N$ be a symbol denoting intersection. This might be less expensive in the memory of a digital computer, but is obviously more computationally expensive when comparing the equivalence of large nested types, as an algorithm that naturally checks for equivalence must walk a tree given by the nested sets of invariants.

When manually searching for invariance, the former method of considering prime numbers is preferable; indeed, any subset of objects that produce a unique value under an operation is sufficient to describe an invariant for some data. This notion gives way to a classification of invariants as an abelian group with prime elements.

\begin{definition}
The abelian group $\mathcal X=(\mathbb X,\cdot)$ is a \textbf{group of invariants}, with $\mathbf X\subset\mathbb X$ being the set of prime elements of $\mathcal X$, which are called the \textbf{initial invariants}. It is useful to use the conventional notation $\chi_a\,\vert\,\chi_b$ to denote $\chi_b=\cdots\cdot\chi_a\cdot\cdots$.
\end{definition}

In the practice of programming on a digital computer, it is necessary to select initial invariants for primitive types, and take the product of these invariants for compound types. It is more involved to find an invariant for a union type that differs from an intersection type, so we find that a combination of the former and latter techniques of picking concrete values for the invariant of a structure is required.

\begin{definition}
The \textbf{structural invariant} $\chi_\varphi$ of a flow $\varphi$ representing an object instantiating a non-initial (either union or intersection) data structure is a pair $(S,x)$ where $S\in\{U,N,P\}$ is a symbol and $x\in\mathcal X$ is an invariant. The symbol $U$ describes a union data structure, $N$ an intersection, and $P$ a primitive (initial) data structure.
\end{definition}

A form of a process lattice for an intersection-type flow, given some invariant $\chi_A\cdot\chi_B\cdot\cdots$ will be something similar (functionally the same) to the following.

\begin{align*}
\Lambda_N(x)&=\begin{pmatrix}
\text{eqpl}(x,\Delta_{\mathbb X}(\chi_A))& a \\
\text{eqpl}(x,\Delta_{\mathbb X}(\chi_B))& b \\
\vdots&\vdots
\end{pmatrix}
\end{align*}

where $x$ is a lifted invariant. In this case, the structure acts to select a value of a type associated with the invariant $\Delta_{\mathbb X}^*(x)$.

As aforementioned, a form of a process lattice for a union-type flow with the same invariant will be something similar to the following.

\begin{align*}
\Lambda_U&=(\:(\chi_A\cdot\chi_B\cdot\cdots)\;\chi_V\;v\:)
\end{align*}

where $\chi_V\,\vert\,\chi_A\cdot\chi_B\cdot\cdots$.

\section{Episodic entropy} \label{sec:Entropy}

Within a system where some flows are treated analagously as data and some as programs, the theory of episodic flows would benefit from a discussion of digital memory. In the general sense, notions of information comparable to information theory follow the structure of flows in the detailed approach to capturing the virtual propagation of information as during the inspection of flows. Within a given flow, it is the thought of success or failure that propagates during inspection. Across a given flow, which is to say from the flow to the result of its inspection, it is the information regarding its original state that propagates to a final value. By parameterising a flow on a variable action, the result of its inspection analgously lends that action to mutation.

To discuss such mutation rigorously, we should observe first the abstract essences of an action: its identity, its meaning, and its content. Formally, discourse revolves around the relations $=,\sim,\vdash$, and the comparison of invariants with $=$. Respectively, the identities (values) of two actions can be compared with $=$, their meanings compared with $\sim$ and $\vdash$, and their contents with the comparison of their structural invariants. Thus, the three entities which uniquely form an action are the action itself, its type, and the actions it composes if it is a flow. Without elementwise comparison of actions, we see it is possible to describe their differences in terms of those three qualities.

It follows that there are five fundamental values that a difference between two actions can take. There is the case when the actions are equal; the case when the actions are unequal but are of the same type and have equal structural invariants; the subsequent case when the actions are of different types but have equal structural invariants; the case when one action precedes another but have unequal structural invariants; and the case when the actions are unequal, neither precede another, and they have unequal structural invariants.

Of these five values, only four assert inequality. These four cases appreciate a natural ordering in the order written. Informally, each case can be reached from the previous by some mutation. An equivalent statement is that no case, without loss of generality, can be reached from the subsequent case. 

\begin{definition} \label{def:DegreesOfLoss}
Given a process lattice $\Lambda^\prime=\overline{\left<\Lambda\setminus\Lambda^{\prime\prime}\right>}$, such that
\begin{align*}
\Lambda&\cong\alpha\\
\Lambda^{\prime}&=(\:\omega^*\;\,\alpha^\prime\:)
\end{align*}
$\alpha$ is said to \textbf{suffer $n$ degrees of entropic loss} under the following conditions:
\begin{align*}
\alpha&\leadsto\alpha^\prime\;\text{ if }\alpha=\alpha^\prime\\
\alpha&\leadsto\dot\gamma\;\text{ if }\alpha\neq\alpha^\prime\text{ but }\alpha\sim\alpha^\prime\text{ and }\chi_\alpha=\chi_{\alpha^\prime}\\
\alpha&\leadsto\ddot\gamma\;\text{ if }\alpha\neq\alpha^\prime\text{ and }\alpha\nsim\alpha^\prime\text{ but }\chi_\alpha=\chi_{\alpha^\prime}\\
\alpha&\leadsto\dddot\gamma\;\text{ if }\alpha\neq\alpha^\prime\text{ and }\chi_\alpha\neq\chi_{\alpha^\prime}\text{ but }\alpha\vdash\alpha^\prime\text{ or }\alpha^\prime\vdash\alpha\\
\alpha&\leadsto\gamma\;\text{ if }\alpha\neq\alpha^\prime,\alpha\nsim\alpha^\prime,\text{ and }\chi_\alpha\neq\chi_{\alpha^\prime}
\end{align*}

where
\begin{itemize}
  \item $\alpha\leadsto\alpha^\prime$ denotes $\alpha$ suffers $0$ degrees of entropic loss,
  \item $\alpha\leadsto\dot\gamma$ denotes $\alpha$ suffers $1$ degree of entropic loss,
  \item $\alpha\leadsto\ddot\gamma$ and $\alpha\leadsto\dddot\gamma$ denote $\alpha$ suffers $2$ degrees of entropic loss,
  \item and $\alpha\leadsto\gamma$ denotes $\alpha$ suffers $3$ degrees of entropic loss.
\end{itemize}

This is supposing $\Lambda^\prime$ is our primary process lattice of discourse. Moreover, if $\alpha\leadsto\beta$, then $\beta$ is the \textbf{episodic entropy} of $\alpha$. The episodic entropy of process lattices can be discussed in the same way.
\end{definition}

We are interested in using the behaviour of episodic entropy in the analysis of variable reference lifetimes in process lattices. To do this, we show that if some reference $P_1$ references garbage data after its encompassing scope is terminated, then it must be the case that another reference $P_2$ references garbage data if $P_2$ shares a scope with $P_1$. In terms of the process lattice calculus, $P_1\leadsto\gamma$ should imply $P_2\leadsto\gamma$, and the general case should hold.

\begin{definition}
A \textbf{state} is a process lattice $\eta=\begin{pmatrix}x_1&\cdots&x_n\end{pmatrix}$, where $x_i$ takes either the value $(\omega^*\;\Lambda_i)$ or $\omega^*$. $\eta$ can also be given by the map $\eta:\mathbb L^n\rightarrow\mathbb L$ such that $\eta(\Lambda_1,\cdots,\Lambda_n)=\begin{pmatrix}x_1&\cdots&x_n\end{pmatrix}$.
\end{definition}

\begin{remark}
In the above definition we do not use single-index notation, but rather we enumerate the process lattices.
\end{remark}

\begin{definition}
An element $\eta_i$ of a state $\eta$ is \textbf{freed through} $\Lambda$ when $\eta^\prime=\begin{pmatrix}
\eta^\prime_1&\cdots&\eta^\prime_{i-1}&\omega^*&\cdots&\eta^\prime_n
\end{pmatrix}=\widehat{\left<\eta\setminus\Lambda\right>}$.
\end{definition}

\begin{definition}
A state $\eta$ is \textbf{freed through} $\Lambda$ when each element $\eta_i$ is freed through $\Lambda$.
\end{definition}

\begin{definition}
A \textbf{reference} is a process lattice $P_\eta$ given by the map $P:\mathbb L\times\mathbb N\rightarrow\mathbb L$, where $P_\eta[i]=\eta_i$. For any reference $P$, there must necessarily be some $\left(\:\omega^*\;\,\alpha\:\right)\sim P$.
\end{definition}

We have discussed a particular composition of process lattices but thus far have not discussed their meanings in an intuitive sense. Consider the description of a process lattice $\Lambda^\prime=\overline{\left<\Lambda\setminus\Lambda^{\prime\prime}\right>}$. Trivially, $\Lambda^\prime$ is a lifted form of the familiar ``result" type $\{\omega^*\}\times\ob{\Phi}$. Considering that $\Lambda\setminus\Lambda^{\prime\prime}$ is a process lattice composed by mapping some lifted result tuple to another process lattice, $\Lambda^{\prime\prime}$ can be thought of as a program in the sense that $\Lambda$ is some input parameter. So $\Lambda^\prime$ is a sort of result considering the inspection of $\Lambda\setminus\Lambda^{\prime\prime}$ to be an execution of the program $\Lambda^{\prime\prime}$ with parameter $\Lambda$.
Similarly, $\widehat{\left<\eta\setminus\Lambda\right>}$ is a state resulting from executing a program parameterised on a state. If $\eta\leadsto\eta$ or $\eta\leadsto\dot\gamma$, it can be said that the state is modified, as its type remains the same (it is still a state). Ergo, a process lattice of the form $\widehat{\left<\eta\setminus\Lambda\right>}\setminus\Lambda^\prime$ should be a program that modifies a state.

By describing a program in this way, we allow for a holistic description of reference lifetimes: a reference $P_\eta$ dies when $P_\eta\leadsto\gamma$, which is the case when $P_\eta$ takes the form $\omega^*$.

\begin{proposition}
Given $\eta^{\prime\prime}=\widehat{\left<\left(\widehat{\left<\eta\setminus\Lambda\right>}\setminus\begin{pmatrix}P_1&\cdots&P_n\end{pmatrix}\right)\right>}$, such that there does not exist an element $\eta_i=\omega^*$ but $\eta$ is freed through $\Lambda$, each reference $P_k\leadsto\gamma$.
\end{proposition}

\begin{proof}
Allow $\eta^\prime=\widehat{\left<\eta\setminus\Lambda\right>}$. Because $\eta$ is freed through $\Lambda$, $\eta^\prime=\begin{pmatrix}
\omega^*&\cdots&\omega^*
\end{pmatrix}$,
$\eta^\prime\setminus\begin{pmatrix}P_1&\cdots&P_n\end{pmatrix}$ is a process lattice of references to the free $\eta^\prime$, such that $\eta^\prime\setminus\begin{pmatrix}P_1&\cdots&P_n\end{pmatrix}=\begin{pmatrix}\omega^*&\cdots&\omega^*\end{pmatrix}$. We see now $\eta^{\prime\prime}=\widehat{\left<\left(\begin{pmatrix}\omega^*&\cdots&\omega^*\end{pmatrix}\right)\right>}=\begin{pmatrix}\omega^*&\cdots&\omega^*\end{pmatrix}$, and so $P_k\leadsto\gamma$ for all $P_k$.
\end{proof}

\section{Conclusion}

The majority of the calculus is described in terms of a specific episodic flow category $\Phi$. This promotes the theory under conventional notions. The theory can, however, be discussed with generality. The utility of basic abstract mathematics would be crucial in conversation of the essential categorical premises, and it is unlikely that novel results would be found--this exploration is left to the reader. However, by specialisation of the framework derived from the initial axioms, we see a simple yet powerful case where flows are treated as simple mathematical objects. 

Certain questions remain; the episodic entropies might form an algebraic structure related to an order $\mathbf 0\leq\dot\gamma\leq\ddot\gamma\leq\dddot\gamma\leq\gamma$, where $\mathbf 0$ denotes no entropic loss, or a less conventional pair of orders $(\mathbf 0\rightleftarrows\gamma,\dot\gamma\leq\ddot\gamma\leq\dddot\gamma)$ to describe the polarity of minimal and maximal loss. Moreover, forming discrete and continuous functions $\bar f:\{\mathbf 0,\dot\gamma,\ddot\gamma,\dddot\gamma,\gamma\}\rightarrow [0,1]$ and $f:\Gamma\rightarrow [0,1]$, with some dense set of analogous loss objects $\Gamma$, which describe proportions of information might allow solving for the Shannon entropy of flows not necessarily in $\Phi$, whether by taking $f$ to be an interpolation of $\bar f$ or otherwise finding an appropriate probability distribution.

The notion of episodic entropy from freed states, in a program, naturally prepares a metric for determining cases in which references to data in memory share a lifetime. This idea further extrapolates to the discussion of information entropy in both classical information science and quantum information science. For example, encoding the trajectory of, say, a photon into a Schwarzschild black hole as a composition of process lattices could allow for analysis of the loss of that information, as might encoding a qubit's measurement.

To describe the actual deletion of information from a universe might require encoding such information as a state which suffers 2 or more degrees of information loss. This method of understanding information entropy is in certain cases more robust than the classic notion of Shannon entropy in information theory, in that the \textit{meaning} of data can be lost just the same as (but independent of) the data itself, as in the case of $\leadsto\ddot\gamma$.

\pagebreak

\bibliographystyle{plain}
\end{document}